\documentclass[%
 reprint, 
superscriptaddress,
 amsmath,amssymb,
 aps, physrev,
longbibliography
]{revtex4-2}

\usepackage{graphicx}    % 图片支持
\usepackage{amsmath, amsthm, amssymb,amsfonts, graphicx, subfigure,cases, threeparttable, booktabs}     % 数学公式
\usepackage{enumitem}
\usepackage{braket}      % 狄拉克符号支持 < | >
\usepackage{tikz}        % 绘
\usetikzlibrary{arrows.meta}
\usepackage{textpos}
\usepackage{cancel} 
\usepackage[normalem]{ulem} 
\usetikzlibrary{calc,shapes.geometric} % Required for ellipse nodes

\newcommand{\supp}{\text{supp}}

\newtheorem{theorem}{Theorem} 

\newtheorem{example}{Example}
\newtheorem{corollary}[theorem]{Corollary}
\newtheorem{lemma}{Lemma}
\newtheorem{proposition}[theorem]{Proposition}
\newtheorem{fact}{Fact}
\theoremstyle{definition} 
\newtheorem{definition}{Definition}

\newcommand{\mc}[1]{\mathcal{#1}}

\newcommand{\old}[1]{\textcolor{red!30}{}}

\begin{document}

% Use the \preprint command to place your local institutional report
% number in the upper righthand corner of the title page in preprint mode.
% Multiple \preprint commands are allowed.
% Use the 'preprintnumbers' class option to override journal defaults
% to display numbers if necessary
%\preprint{}

%Title of paper
\title{Quantum states supported by matroids}

% repeat the \author .. \affiliation  etc. as needed
% \email, \thanks, \homepage, \altaffiliation all apply to the current
% author. Explanatory text should go in the []'s, actual e-mail
% address or url should go in the {}'s for \email and \homepage.
% Please use the appropriate macro foreach each type of information

% \affiliation command applies to all authors since the last
% \affiliation command. The \affiliation command should follow the
% other information
% \affiliation can be followed by \email, \homepage, \thanks as well.
\author{Xiaowei Huang}
\email{huangxiaowei@gziis.org}
\affiliation{Institute of Intelligent Software, Guangzhou 325035, China}
\affiliation{Institute of Quantum Computing and Software, School of Computer Science and Engineering, Sun Yat-sen University, Guangzhou 510006, China}

\author{Fei Shi}
\email{shif26@mail.sysu.edu.cn}
\affiliation{Institute of Quantum Computing and Software, School of Computer Science and Engineering, Sun Yat-sen University, Guangzhou 510006, China}

\author{Lijun Zhang}
\email{zhanglj@ios.ac.cn}
\affiliation{Institute of Intelligent Software, Guangzhou 325035, China}
\affiliation{Key Laboratory of System Software of Chinese Academy of Sciences, Institute of Software Chinese Academy of Sciences, University of Chinese Academy of Sciences, Beijing 100045, China}

\author{Lvzhou Li}
\email{lilvzh@mail.sysu.edu.cn (Corresponding author)}
\affiliation{Institute of Quantum Computing and Software, School of Computer Science and Engineering, Sun Yat-sen University, Guangzhou 510006, China}
%\homepage[]{Your web page}
%\thanks{}
%\altaffiliation{}
%\affiliation{}

%Collaboration name if desired (requires use of superscriptaddress
%option in \documentclass). \noaffiliation is required (may also be
%used with the \author command).
%\collaboration can be followed by \email, \homepage, \thanks as well.
%\collaboration{}
%\noaffiliation

\date{\today}

\begin{abstract}

%In this work, we establish a structural correspondence between quantum states and matroid theory, revealing that fundamental notions such as entanglement and measurement can be characterized combinatorially through matroids—a concept seemingly unrelated to quantum information. 

In this work, we establish a structural correspondence between quantum states and matroid theory. This connection demonstrates that key properties of quantum states, including entanglement and measurement, can be characterized in purely combinatorial terms via matroids, despite the apparent conceptual distance between these two fields. Using this framework, we show that a matroid-supported state is genuinely entangled when its underlying matroid is connected. Moreover, a uniform superposition over all bases of a matroid is genuinely entangled if and only if the matroid is connected. We also demonstrate that a local measurement in the $Z$-basis on such a state yields another matroid-supported state, whose underlying matroid is a minor of the original one. Inspired by matroid duality, we further propose a notion of quantum state duality, uncovering a deep structural symmetry in state transformations.
\end{abstract}

% insert suggested keywords - APS authors don't need to do this
%\keywords{Quantum Entanglement, Genuine Multipartite Entanglement, Matroid, Connectivity, Multiaffine Polynomial}

%\maketitle must follow title, authors, abstract, and keywords
\maketitle

% body of paper here - Use proper section commands
% References should be done using the \cite, \ref, and \label commands
\section{Introduction}
Quantum entanglement—first brought to light in the seminal work of Einstein, Podolsky, and Rosen (EPR)~\cite{einstein1935can} in 1935—was initially discussed to highlight what they perceived as the incompleteness of quantum mechanics. Shortly thereafter, Schr\"{o}dinger~\cite{schrodinger1935discussion} coined the term Verschr\"{a}nkung (entanglement), identifying it as a defining feature of quantum theory. Decades later, Bell~\cite{bell1964einstein} formulated his celebrated inequalities, providing an experimentally testable criterion to distinguish between local hidden-variable theories and quantum mechanical predictions. Groundbreaking experiments, beginning with those of Aspect, Grangier, and Roger~\cite{aspect1982experimental} in the 1980s, confirmed the violation of Bell inequalities, firmly establishing entanglement as a physical phenomenon rather than a mere mathematical curiosity.

Today, research on quantum entanglement has become extensive and profound, encompassing its theoretical foundations, methods for characterization, and quantitative measures~\cite{horodecki2009quantum, bengtsson2017geometry}. Among the various forms of entanglement, genuine multipartite entanglement (GME) stands out as a particularly robust and useful form for quantum information processing. GME serves as a key resource in quantum computation~\cite{raussendorf2003measurement,zhao2025entanglement}, quantum communication~\cite{karlsson1998quantum,yin2020entanglement}, quantum error correction~\cite{makuta2021self,makuta2025all}, and quantum metrology~\cite{chin2012quantum,demkowicz2014using}. Recently, GME has been demonstrated in systems with up to 51 superconducting qubits, demonstrating the potential for realizing GME in even larger-scale systems~\cite{cao2023generation}. However, determining whether a given state has GME remains a fundamental and challenging problem in the field.

This work is motivated by the following questions.

\textbf{Question 1.} Some important quantum states, such as Dicke states, graph states, and hypergraph states, are all related to combinatorial structures and have significant applications. Are there other combinatorial structures worth investigating?

\textbf{Question 2.} What kind of combinatorial structure makes the entanglement of quantum states independent of their non-zero coefficients?

We provide a combinatorial explanation for several key quantum phenomena—such as entanglement and measurement—exhibited by a class of special quantum states, namely the matroid-supported states (see Definition~\ref{def:matroid-supported-state}) introduced herein. The matroid picture offers an elegant way of characterizing the multipartite entanglement of matroid-supported states. %Todo(Typos and xxx)::\remove{Thus, the connectivity of the matroid can used to detect GME.}
 
Matroid theory, originating from the pioneering works of Whitney~\cite{whitney1992abstract} and Nakasawa~\cite{nakasawa1935axiomatik}, serves as a fundamental framework in combinatorics and theoretical computer science. It abstracts essential properties such as linear independence in vector spaces and acyclicity in graphs into an elegant axiomatic structure. This abstraction not only unifies various independence phenomena but also provides a powerful combinatorial language for analyzing a wide range of problems within a unified conceptual framework. Owing to their rich structure and broad applicability, matroids have found deep connections with optimization, graph theory, coding theory, and more recently quantum information theory. 

\subsection{Results}
In this paper, we investigate quantum states through the lens of their underlying combinatorial structure---namely, matroids. We demonstrate that quantum entanglement and quantum measurement can be understood not merely algebraically but also combinatorially, by interpreting a state’s support as a collection of matroid bases (see Definition~\ref{def:matroid-supported-state} and Definition~\ref{def:matroid-base-state}). This perspective reveals a fundamental correspondence between the combinatorial geometry of a state and its GME, as well as between measurement and matroid minors. Our main results are as follows.
\begin{enumerate}
\item A matroid-supported state is genuinely entangled whenever its underlying matroid is connected. (See Theorem~\ref{th:matroid-supported-state})

\item A uniform superposition state over all bases of a matroid is genuinely entangled if and only if the corresponding matroid is connected. (See Theorem~\ref{th:matroid-base-state})

\item A local measurement in the $Z$-basis on a matroid-supported state yields another matroid-supported state whose underlying matroid is precisely a minor of the original matroid. (See Theorem~\ref{th:meaurement-minor})
\end{enumerate}

Furthermore, inspired by matroid duality, we introduce a notion of quantum state duality, uncovering a deep structural symmetry in state transformations. This dual perspective highlights how matroid operations naturally encode transformations among entangled quantum systems.

\begin{enumerate}[resume]
\item A quantum state is genuinely entangled if and only if its dual state is genuinely entangled. (See Proposition~\ref{le:dual-state-engtangled})

\item The dual state of a matroid-supported state is a matroid-supported state supported by the dual matroid of the original matroid. (See Proposition~\ref{le:dual-matroid-supported-state} and Proposition~\ref{le:dual-matroid-base-state})
\end{enumerate}

\subsection{Technique}
\textbf{Entanglement, Irreducibility, and Connectedness}.
The notions of quantum entanglement, polynomial irreducibility, and matroid connectedness share a unifying theme — indecomposability. Each concept captures, in its respective domain, the idea that a structure cannot be decomposed into simpler, independent components without losing its essential properties. From physical, algebraic, and combinatorial perspectives, these notions reflect different manifestations of dependence and connectivity within complex systems. Their mathematical forms are consistent.
\begin{center}
\begin{itemize}
\item Quantum entanglement: $|\psi\rangle$ {$\neq$} $|\psi_{1}\rangle \otimes |\psi_{2}\rangle$ 
\item Irreducible polynomial: $p(x)$ $\neq$  $p_1(x)\cdot p_2(x)$
\item Connected matroid: $M$ {$\neq$} $M_1\oplus M_2$
\end{itemize}
\end{center}

\noindent{}\textbf{Technique overview}.
To prove Theorem~\ref{th:matroid-supported-state} and Theorem~\ref{th:matroid-base-state}, we adopt an algebraic representation of a quantum state by associating it with a multiaffine polynomial, called generating polynomial, whose coefficients encode the amplitudes of the state. %The entanglement of the quantum state can then be characterized by the irreducibility of its generating polynomial A quantum state is genuinely entangled if and only if its generating polynomialis irreducible
Then, we will prove a key lemma: A quantum state is genuinely entangled if and only if its generating polynomial is irreducible (see Lemma~\ref{le:entanglement-by-irreducibility}).

Let $f(x)=\sum_{B\subseteq[n]}\alpha_B x^B$ be a homogeneous multiaffine polynomial associated with a matroid, in the sense that its support $\supp(f)$ is precisely the collection of bases of a matroid. The irreducibility of $f$ is intimately connected to the connectedness of the underlying matroid. Specifically, as demonstrated in Lemma~\ref{le:conncted-matroid-to-irreducible}, the connectedness of the matroid provides a sufficient condition for the irreducibility of the corresponding polynomial.

The overall logical structure of the proofs can thus be summarized as Figure~\ref{fig:proof-overview}.

%\begin{turnpage}
\begin{figure}[htbp]
  \centering
\begin{tikzpicture}[
    set/.style={draw, rounded corners=8pt, minimum width=4cm,
      align=center},scale=0.5
]

\node[set] (IP) at (0,0) {Irreducible polynomial};
\node[set,below of=IP] (GM) {Genuinely entangled state};
\node[set,above of=IP] (CM) {Connected matroid};

\draw[double, double distance=2pt,-{Implies}] (CM) -- (IP);
\draw[double, double distance=2pt, {Implies}-{Implies}] (GM) -- (IP);
\end{tikzpicture}    
  \caption{The proof framework for Theorem~\ref{th:matroid-supported-state} and Theorem~\ref{th:matroid-base-state}.}
  \label{fig:proof-overview}
\end{figure}
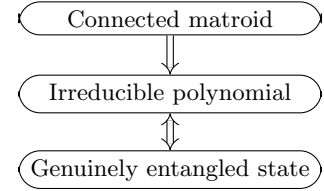
%\end{turnpage}

This correspondence establishes a deep structural analogy among combinatorial, algebraic, and quantum domains, showing that entanglement in quantum mechanics, irreducibility in algebra, and connectivity in matroid theory all express a common principle of inseparability.

\subsection{Related works}
%\noindent{}\textbf{Matroids in quantum information and quantum computing}.
Matroid is the core concept of combinatorial optimization and has wide applications in algorithm design, algorithm analysis, coding, geometry, and other fields. It also inevitably appears in quantum information and quantum computing and is gradually gaining attention. These applications include computation models~\cite{shepherd2009temporally,shepherd2010binary}, 
quantum simulating~\cite{mann2021simulating},
%error-correcting codes~\cite{sarvepalli2010local,sarvepalli2014quantum},
%(Sarvepalli and Raussendorf~\cite{sarvepalli2010local} and Sarvepalli~\cite{sarvepalli2014quantum}), 
secret sharing~\cite{sarvepalli2010matroids} and
%(Sarvepalli and Raussendorf~\cite{sarvepalli2010matroids}) 
circuit optimization~\cite{amy2014polynomial}.
%(Amy, Maslov and Mosca~\cite{amy2014polynomial}).
At the same time, quantum algorithms for matroid problems have also attracted the attention of relevant researchers~\cite{kulkarni2013query,huang2024quantum-tcs,huang2024quantum-fcs,huang2023quantum}.

In addition, Gurvits~\cite{gurvits2003classical,gurvits2004classical} introduced a quantum generalization of classical matching theory, a fundamental concept in matroid theory, to investigate the structure of quantum entanglement. In his seminal work, Gurvits demonstrated that for separable bipartite quantum states, the existence of a perfect quantum matching is equivalent to the full-rank property of the intersection of two geometric matroids. This approach establishes a profound link between matroid intersection theory and the separability problem in quantum information. In contrast, our work explores the entanglement of quantum states through a different structural lens—by characterizing GME via the connectivity of matroids rather than matroid intersections.

Another promising research direction originates from the connection between matroid theory and quantum codes, as explored by Sarvepalli~\cite{sarvepalli2010local,sarvepalli2014quantum}. This approach offers a unifying framework, which is most evident in the fact that all stabilizer states are supported by symplectic matroids. This generalization of classical matroid theory provides a powerful combinatorial lens for analyzing quantum states and may lead to new tools for code construction and classification.

\subsection{Organization}
The remainder of this paper is organized as follows. In Section~\ref{pre}, we define the notations, provide the background knowledge of quantum states and the necessary concepts from matroid theory. In Section~\ref{main-section}, we first give the definition of affine homogeneous polynomial, which is the core tool of our proofs. Then we give the proof of our results. Finally, we conclude this article in Section~\ref{conclusion}.

\section{Preliminaries}\label{pre}
\noindent\textbf{Notations}: 
Let $E$ be a finite set, $A$ and $B$ be the subsets of $E$, and $x\in E$. We define $B\backslash{}A:=\{x\in E|x\in B\;\text{and}\; x\notin A\}$. We use $|E|$, $\overline{A}$, $A+x$, and $A-x$ to denote the cardinality of $E$, the complement of $A$, the set $A\cup\{x\}$, and the set $A\backslash{}\{x\}$, respectively. We denote by $2^E$ the set of all subsets of $E$.
For a positive integer $n$, we use $[n]$ to denote the set $\{1,\dots,n\}$. For a string $x\in\{0,1\}^n, i\in[n]$, $x_i$ denotes the $i$-th bit of $x$. Label an $n$-qubit system by the elements of $[n]$. Then for each $e\in [n]$, the notation $|0_e\rangle$ and $|1_e\rangle$ indicates that the qubit labeled by $e$ is in the state $|0\rangle$ and $|1\rangle$, respectively. $\mathbb{Z}$  represents the set of integers. Let $x=\{x_i\}_{i\in[n]}$ be $n$ variables and $S\subseteq [n]$, $x^S=\prod_{i\in S}x_i$. For disjoint sets $A$ and $B$, we use $A\sqcup B$ to denote their disjoint union. 
\subsection{Backgrounds}

Let $|\psi\rangle=\sum_{S\subseteq [n]}\alpha_S|S\rangle$ be a  quantum state over an $n$-qubit system, where $|S\rangle$ denotes the computational basis state in which the $i$-th qubit is in state $|1\rangle$ if $i\in S$ and $|0\rangle$ otherwise, and $[n]=\{1,2,\cdots,n\}$.

\begin{definition}[Genuine multipartite entanglement]
A quantum state $|\psi\rangle$ is called \emph{biseparable} if there exists a non-empty proper subset $T\subsetneq [n]$ such that
\begin{equation}
 |\psi\rangle=|\psi_{T}\rangle\otimes|\psi_{[n]\backslash{}T}\rangle,     
\end{equation} 
where $|\psi_T\rangle$ and $|\psi_{[n]\backslash{T}}\rangle$ are quantum states over the qubits labeled by $T$ and $[n]\backslash{T}$, respectively. Otherwise, it is called \emph{genuinely entangled}. 
\end{definition}

\begin{definition}
The support of a quantum state $|\psi\rangle=\sum_{S\subseteq [n]}\alpha_S|S\rangle$ is defined as
\begin{equation*}
    \supp(|\psi\rangle):=\{S\subseteq[n]\;\big|\;\alpha_S\neq 0\}.
\end{equation*}
\end{definition}

\begin{example}\label{ex:support-set-state}
The well-known genuinely entangled states in quantum information and quantum computing are as follows:
\begin{enumerate}
\item Bell states: $|\Phi^\pm\rangle=\frac{1}{\sqrt{2}}(|00\rangle\pm|11\rangle)$, $|\Psi^\pm\rangle=\frac{1}{\sqrt{2}}(|01\rangle\pm|10\rangle)$.

\item GHZ states: $|\text{GHZ}_n\rangle=\frac{1}{\sqrt{2}}(|0\rangle^{\otimes n}+|1\rangle^{\otimes n})$ with $n\geq 3$ being an integer.

\item W states: $|W_n\rangle = \frac{1}{\sqrt{n}}\sum_{i=1}^{n}|0\cdots 01_i0\cdots0\rangle$ with $n\geq 3$ being an integer, where $1_i$ denotes that the $i$-th qubit is in state $1$.

\item Dicke states: $|D_k^n\rangle=\frac{1}{\sqrt{\binom{n}{k}}}\sum_{S\subseteq [n]:|S|=k}|S\rangle$  with $n, k$ being positive integers and $1\leq k\leq n-1$.
\end{enumerate}
Their supports are:
\begin{enumerate}
\item  $\supp(|\Phi^\pm\rangle) = \{\emptyset,\{1,2\}\}$, $\;\supp(|\Psi^\pm\rangle) = \{\{1\},\{2\}\}$.
\item $\supp(|\text{GHZ}_n\rangle)=\{\emptyset,\{1,2,\dots,n\}\}$.
\item $\supp(|W_n\rangle)=\{\{i\}:i\in[n]\}$.
\item $\supp(|D_k^n\rangle)=\{S\subseteq[n]:|S|=k\}$.
\end{enumerate}
\end{example}

GME is a subtle and multifaceted phenomenon, whose mathematical characterization reveals a complex interplay between amplitude coefficients and their underlying structure. When the effect of these coefficients is disregarded, GME typically reflects the combinatorial  constraints inherent in quantum systems. This perspective naturally raises the question: What kind of nontrivial combinatorial structures can be used to explain that quantum states are genuinely entangled?

In Example~\ref{ex:support-set-state}, these representative and well-studied genuinely entangled states exhibit two combinatorial structures such that their GME is independent of their coefficients. We observe that apart from the Dicke states, the supports of the  quantum states correspond to partitions of $[n]$. This partition structure is self‑evident and provides an intuitive explanation for the genuine entanglement of the state. Its genuine entanglement is independent of the specific amplitude coefficients.
Here we allow the partition to have an empty set (which represents state $|0\rangle^{\otimes n}$).
Based on a partition, we can easily construct a new structure by adding some new subsets to the partition so that the GME of a quantum state with this underlying structure is still independent of the coefficients. However, these structures are a bit ``arbitrary" and difficult to characterize accurately, so this is not the structure we are concerned with. Here is an example illustrating this ``arbitrary" structure.

\begin{example}
Let $\mc{P}=\{\{1,2\},\{3,4\},\{5\}\}$ be a partition of $\{1,2,3,4,5\}$, and let $S=\{1,3,5\}$ be a subset of $\{1,2,3,4,5\}$ that is not an element of $\mc{P}$. Any quantum state with support $\mc{P}\cup \{S\}$ is genuinely entangled.
\end{example}

On the other hand, in Example~\ref{ex:support-set-state}, from the expressions of these quantum states $|\Psi^\pm\rangle(=\frac{1}{\sqrt{2}}(|01\rangle\pm|10\rangle))$, $|W_n\rangle$ and $|D^n_k\rangle$, we can see that each element in their support is equal in size. This suggests a subtler combinatorial structure, \emph{matroid}, supporting these quantum states. We first define a class of quantum states related to matroids, which we call matroid-supported states.

\begin{definition}[Matroid-supported state]\label{def:matroid-supported-state}
A quantum state $|\psi\rangle=\sum_{S\subseteq {[n]}}\alpha_S|S\rangle$ is called a \emph{matroid-supported state} if its support $\supp(|\psi\rangle)$ is the collection of bases of a matroid $M$ on the ground set $[n]$, where  $M$ is called its supporting matroid.
\end{definition}

At the same time, motivated by two well-studied classes of quantum states, graph states~\cite{hein2004multiparty} %(Hein, Eisert and Briegel~\cite{hein2004multiparty}) 
and hypergraph states~\cite{rossi2013quantum}, %(Rossi, Huber, Bru{\ss} and Macchiavello~\cite{rossi2013quantum}), 
we consider, for a given matroid $M=([n],\mc{B})$, the quantum state defined as the uniform superposition over all bases of the given matroid $M$. This is a special type of matroid-supported state. This construction naturally encodes the combinatorial structure of a matroid into a quantum state called matroid-base state. Its definition is as follows. 

\begin{definition}[Matroid-base state]\label{def:matroid-base-state}
Given a matroid $M=([n],\mc{B})$, the quantum state defined by
\begin{equation}
 |M\rangle=\frac{1}{\sqrt{|\mc{B|}}}\sum_{B\in\mc{B}}|B\rangle ,   
\end{equation}
is called a \emph{matroid-base state} associated to matroid $M$.
\end{definition}

Formally, there is a  well-known state called Bethe state, which includes the matroid-supported states. The Bethe state~\cite{bethe1997theory}, first introduced in 1931 to solve the Heisenberg model, is an exact many-body state describing $M$ excitations on an $L$-site lattice. It can be written as:
\begin{equation}
|\psi\rangle=\sum_{S\subseteq 2^{[L]}:|S|=M}\alpha_S|S\rangle
\end{equation}
where $S$ represents the location of the excitation of size $M$, and the coefficients $\alpha_S$ are constructed via the Bethe Ansatz to satisfy the integrability constraints of the system.

The relationship between the quantum states studied in this paper and the Dicke, Bethe, and GME states is shown in Figure~\ref{fig:bethe-states}.

This paper does not address how to determine whether a given state is a matroid-supported state. Throughout this paper, we assume that we already know it is a matroid-supported state.

%\begin{turnpage}
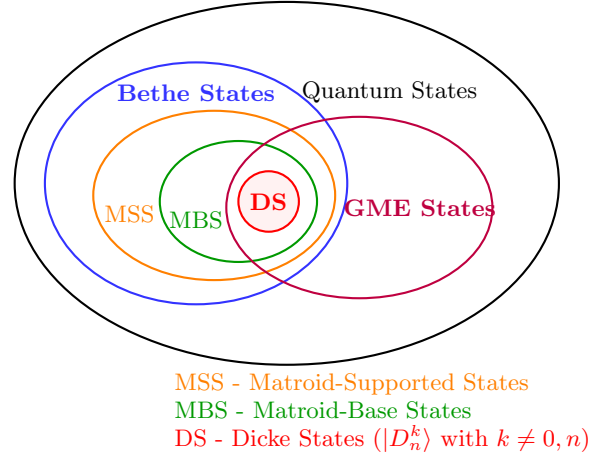
\begin{figure}[htbp]
  \centering
\begin{tikzpicture}[scale=0.8]
\draw[thick] (0,0) ellipse (4.5cm and 3cm);
\node[] at (1.7, 1.5) {Quantum States};

\draw[thick, blue!80] (-1.5, 0) ellipse (2.5cm and 2cm);
\node[blue!80] at (-1.5, 1.5) {\textbf{Bethe States}};

\draw[thick, orange] (-1.2, -0.2) ellipse (2cm and 1.4cm);
\node[orange] at (-2.6, -0.5) {MSS};

\draw[thick, green!60!black] (-0.8, -0.3) ellipse (1.3cm and 1cm);
\node[green!60!black] at (-1.5, -0.6) {MBS};

\draw[thick, red, fill=red!10, fill opacity=0.5] (-0.3, -0.3) ellipse (0.5cm and 0.5cm);
\node[red, font=\bfseries] at (-0.3, -0.3) {DS};

\draw[thick, purple] (1.2, -0.4) ellipse (2.2cm and 1.5cm);
\node[purple] at (2.2, -0.4) {\textbf{GME States}};

\node[anchor=north west, align=left] at (-2, -3) {
    \textcolor{orange}{MSS - Matroid-Supported States}\\
    \textcolor{green!60!black}{MBS - Matroid-Base States}\\
    \textcolor{red}{DS - Dicke States ($|D^k_n\rangle$ with $k\neq 0,n$)}
};
\end{tikzpicture}
  \caption{The classification of quantum states discussed in this article. Their inclusion relationships are shown in the figure. By Corollary~\ref{co:dicke}, Dicke states $|D^k_n\rangle$ with $k\neq 0,n$ are genuinely entangled. As can be seen from Example~\ref{ex:disconnected-matroid-shows-GE}, there exist matroid-base states that are not GME. 
  }
  \label{fig:bethe-states}
\end{figure}
%\end{turnpage}

\subsection{Matroid Theory}\label{sec:matroid}
Here we give some basic definitions and concepts on matroids. 
Matroid theory is established as a generalization of linear algebra and graph
theory. Some concepts are similar to those of linear algebra or graphs.
We refer the reader to Oxley~\cite{OUP/oxley11} and Welsh~\cite{welsh1976matroid} for more details about matroid theory.

\begin{definition}[\textbf{Matroid}]
  \label{def:matroid}
  A \emph{matroid} is a combinational object defined by the tuple
  $M=([n],\mathcal{I})$ on the finite ground set $[n]$ and $\mathcal{I}\subseteq 2^{[n]}$ such
  that the following properties hold:
  \begin{enumerate}
  \item[\textbf{I0}.] $\emptyset \in \mathcal{I}$;
  \item[\textbf{I1}.] If $A'\subseteq A$ and $A\in\mathcal{I}$, then $A'\in\mathcal{I}$;
  \item[\textbf{I2}.] For any two sets $A,B\in\mathcal{I}$ with $|A|<|B|$, there 
  exists an element $x\in B\backslash{A}$ such that $A+x\in\mathcal{I}$.
  \end{enumerate}
\end{definition}
The members of $\mathcal{I}$ are called \emph{independent} set of $M$.
The subsets of $[n]$ not belonging to $\mathcal{I}$ are called \emph{dependent}.
A \emph{base} of $M$ is a \emph{maximal independent} set, and the collection of bases is denoted by $\mc{B}(M)$ or $\mc{B}$. The property \textbf{I2} implies that the elements of $\mc{B}$ have the same cardinality, which is the rank of $M$ (see Definition~\ref{def:rank}). A \emph{circuit} of $M$ is a \emph{minimal dependent} set, and the collection of circuits is denoted by $\mc{C}(M)$ or $\mc{C}$.
 
\begin{example} Three well-known matroids.
\begin{itemize}
    \item \textbf{Uniform matroid}. For integer $r$ with $0\leq r\leq n$, let $U_{r,n}=([n],\mc{I})$, where $\mc{I}=\{I\subseteq [n]\;|\;|I|\leq r\}$. $U_{r,n}$ is called \emph{uniform matroid}.
    \item \textbf{Vector matroid}.  Let $\mathbb{F}$ be a field and $A\in\mathbb{F}^{m\times n}$ be a matrix over $\mathbb{F}$. Let 
  \begin{align*}\mc{I}=&\{I\subseteq[n]\;|\;\text{The column vectors of}\;A\;  \text{indexed by}\;I~\\
  &\text{are linearly independent over}~ \mathbb{F}\}.  
  \end{align*}
    Then $M[A]=([n],\mc{I})$ is a matroid, called \emph{vector matroid}.
    \item \textbf{Graphic matroid}. Let $G=(V,E)$ be an undirected graph. Let
    \begin{equation*}
        \mc{I}=\{I\subseteq E\;|\;I\;\textbf{is a forest of } G.\}
    \end{equation*}
    Then $M[G]=(E,\mc{I})$ is a matroid, called \emph{graphic matroid}.
\end{itemize}
\end{example}

\begin{definition}[\textbf{Rank}]
  \label{def:rank}
  The \emph{rank function} of a matroid $M=([n],\mathcal{I})$ is the function $r:2^{[n]}\rightarrow\mathbb{Z}$ defined by
  \begin{equation}
   r(A) = \max\{|X|\;\big|\;X\subseteq A,X\in\mathcal{I}\}\;\;\;\;\;\;(A\subseteq [n]).
   \end{equation}
\noindent{}The rank of $M$, denoted by $r(M)$, is $r([n])$. 
\end{definition}
For any $e\in[n]$ with $r(\{e\}) = 0$, we say that $e$ is a \emph{loop} of $M$. For any two non-loop elements $e,f\in[n]$ with $r(\{e,f\})=1$, we say that $e,f$ are \emph{parallel}. A matroid is called \emph{simple matroid} if it does not have loops or parallel elements.

\begin{lemma}[Base Axioms]\label{th:base-axioms}
A non-empty collection $\mc{B}$ of subsets of $[n]$ is the collection of bases of a matroid on $[n]$ if and only if it satisfies the following condition:
\begin{itemize}
\item[(B1)] If $B_1,B_2\in\mathcal{B}$ and $x\in B_{1}\backslash{}B_{2}$, then there exists $y\in B_{2}\backslash{}B_{1}$ such that $B_{1}+y-x\in\mathcal{B}$.\label{B1}
\end{itemize}
Such a matroid determined by the collection of bases $\mc{B}$ is usually represented by $([n],\mc{B})$.
\end{lemma}

\begin{lemma}[Duality]\label{le:matroid-duality}
Let $M=([n],\mc{B})$ be a matroid and $\mc{B}^*$ be $\{[n]\backslash{}B\;\big|B\in\mc{B}(M)\;\}$. Then $\mc{B}^*$ is the collection of bases of a matroid $M^*$ on $[n]$.
\end{lemma}
We call $M^*$ the dual matroid of $M$. Obviously, the relationship between $M$ and $M^*$ is symmetrical, that is,
\begin{equation*}
    (M^*)^* = M.
\end{equation*}

The rank function of $M^*$ is called the \emph{corank} function of $M$. A \emph{cobase} of $M$ is a base of $M^*$, a \emph{cocircuit} of $M$ is a circuit of $M^*$, if $e$ is a loop of $M^*$ it is called a \emph{coloop} of $M$ and so on.

The proofs of Base Axioms and Duality can be found in \cite{OUP/oxley11,welsh1976matroid}.

\begin{definition}[\textbf{Connected Matroid}]
\label{de:non-separable-matroid}
A matroid $M = ([n],\mc{I})$ with rank function $r$ is called \emph{connected} or \emph{non-separable} if every
nonempty proper subset $A\subsetneq [n]$ satisfies
\begin{equation}
r(A)+r(E-A)>r(E). 
\end{equation}
Otherwise $M$ is called \emph{disconnected} or \emph{separable}.
\end{definition}

\begin{lemma}\label{le:connected-matroid}
A matroid $M = ([n],\mc{B})$ is \emph{connected} if and only if for every pair of distinct elements $x$ and $y$ of $[n]$, there is a circuit of $M$ containing $x$ and $y$.
\end{lemma}

\begin{lemma}\label{le:connected-matroid-sum}
A matroid $M=([n],\mc{B})$ is \emph{connected} if and only if there is no proper not-empty subset $T\subsetneq[n]$ such that $M$ has a \emph{direct sum} or \emph{1-sum} expression: $M=M_1\oplus M_2$, where $M_1=(T,\mc{B}_1)$ and $M_2=([n]\backslash{}T,\mc{B}_2)$ are matroids and the $\oplus$ operation means that $\mc{B}=\{B_1\sqcup B_2\;|\;B_1\in\mc{B}_1\;\text{and}\;B_2\in\mc{B}_2\}$.
\end{lemma}

\begin{lemma}\label{le:connected-matroid-dual}
A matroid $M = ([n],\mc{B})$ is \emph{connected} if and only if its dual matroid $M^*$ is connected.
\end{lemma}

The proofs of Lemma~\ref{le:connected-matroid}, Lemma~\ref{le:connected-matroid-sum} and Lemma~\ref{le:connected-matroid-dual} can be found in \cite{OUP/oxley11,welsh1976matroid}.

\begin{definition}[\textbf{Deletion}]
Let $M=([n],\mc{B})$ be a matroid and $e\in[n]$. The \emph{deletion} $M\backslash{e}$ of $e$ from $M$ is the matroid $([n]-e, \mc{B}(M\backslash{e}))$, where
\begin{equation}
    \mc{B}(M\backslash{e}) = \left\{ 
    \begin{array}{l}
    \mc{B}(M)\hfill\text{if $e$ is a loop},\\
    \{B-e\;|\;B\in\mc{B}(M)\}\hfill\text{\;\;\;\;if $e$ is a coloop},\\
    \{B\in\mc{B}(M)\;|\;e\notin B\}\hfill\text{otherwise}.
    \end{array}
    \right.
\end{equation}
\end{definition}

\begin{definition}[\textbf{Contraction}]
Let $M=([n],\mc{B})$ be a matroid and $e\in[n]$. The \emph{contraction} $M/{e}$ of $e$ from $M$ is the matroid $([n]-e, \mc{B}(M/{e}))$, where
\begin{equation}
    \mc{B}(M/{e}) = \left\{ 
    \begin{array}{l}
    \mc{B}(M)\hfill\text{if $e$ is a loop},\\
    \{B-e\;|\;B\in\mc{B}(M)\}\hfill\text{if $e$ is a coloop},\\
    \{B-e\;|\;B\in\mc{B}(M)\;\text{and}\;e\in B\}\hfill\text{\;otherwise}.
    \end{array}
    \right.
\end{equation}
\end{definition}

It is not difficult to check that the deletion and contraction do actually give matroids. The operations behave well, so, for example, they commute with each other and with themselves. 

\begin{example}
An example illustrating the concepts of connectivity, deletion, contraction, and duality, all of which are extended from graph theory. See Figure~\ref{fig:matroid-concepts}.
\end{example}

\begin{figure}[htbp]
\centering
\begin{tikzpicture}%[every node/.style={circle, fill=black,inner sep=0pt,minimum size=3pt}]
% Graph G
\node[circle, fill=black,inner sep=0pt,minimum size=3pt,label=below:{$A$}] (A) at (0,0) {};
\node[circle, fill=black,inner sep=0pt,minimum size=3pt,label=below:{$B$}] (B) at ($(A)+(2.4,0)$) {};
\node[circle, fill=black,inner sep=0pt,minimum size=3pt,label=below:{$X$}] (X) at ($(A)+(1.2,0.5)$) {};
\node[circle, fill=black,inner sep=0pt,minimum size=3pt,label=below:{$Y$}] (Y) at ($(X)+(0,0.7)$) {};
\node[circle, fill=black,inner sep=0pt,minimum size=3pt,label=below:{$Z$}] (Z) at ($(Y)+(0,0.7)$) {};

\draw (A) -- (X) -- (B);
\draw (A) -- (Y) -- (B);
\draw (A) -- (Z) -- (B);
\node[] at ($(X)+(0, -1.2)$) {$G=(V,E)$};
\node[] at ($(X)+(0, -1.8)$) {$(a)$};

% Graph G\(B,X)
\node[circle, fill=black,inner sep=0pt,minimum size=3pt,label=below:{$A$}] (A1) at (3.3,0) {};
\node[circle, fill=black,inner sep=0pt,minimum size=3pt,label=below:{$B$}] (B1) at ($(A1)+(2.4,0)$) {};
\node[circle, fill=black,inner sep=0pt,minimum size=3pt,label=below:{$X$}] (X1) at ($(A1)+(1.2,0.5)$) {};
\node[circle, fill=black,inner sep=0pt,minimum size=3pt,label=below:{$Y$}] (Y1) at ($(X1)+(0,0.7)$) {};
\node[circle, fill=black,inner sep=0pt,minimum size=3pt,label=below:{$Z$}] (Z1) at ($(Y1)+(0,0.7)$) {};

\draw (A1) -- (X1);% -- (B1);
\draw (A1) -- (Y1) -- (B1);
\draw (A1) -- (Z1) -- (B1);
\node[] at ($(X1)+(0, -1.2)$) {$G\backslash{(B,X)}$};
\node[] at ($(X1)+(0, -1.8)$) {$(b)$};

% Graph G/(B,X)
\node[circle, fill=black,inner sep=0pt,minimum size=3pt,label=below:{$A$}] (A2) at (0,-4) {};
\node[circle, fill=black,inner sep=0pt,minimum size=3pt,label=below:{$W$}] (B2) at ($(A2)+(2.4,0)$) {};
%\node[circle, fill=black,inner sep=0pt,minimum size=3pt,label=below:{$X$}] (X2) at ($(A2)+(1.2,0.5)$) {};
\coordinate (X2) at ($(A2)+(1.2,0.5)$);
\node[circle, fill=black,inner sep=0pt,minimum size=3pt,label=below:{$Y$}] (Y2) at ($(X2)+(0,0.7)$) {};
\node[circle, fill=black,inner sep=0pt,minimum size=3pt,label=below:{$Z$}] (Z2) at ($(Y2)+(0,0.7)$) {};

\draw (A2) -- (B2);
\draw (A2) -- (Y2) -- (B2);
\draw (A2) -- (Z2) -- (B2);
\node[] at ($(X2)+(0, -1.2)$) {$G/(B,X)$};
\node[] at ($(X2)+(0, -1.8)$) {$(c)$};

% Graph G*
\node[circle, fill=gray,inner sep=0pt,minimum size=3pt] (A4) at (3.3,-4) {};
\node[circle, fill=gray,inner sep=0pt,minimum size=3pt] (B4) at ($(A4)+(2.4,0)$) {};
\node[circle, fill=gray,inner sep=0pt,minimum size=3pt] (X4) at ($(A4)+(1.2,0.5)$) {};
\node[circle, fill=gray,inner sep=0pt,minimum size=3pt] (Y4) at ($(X4)+(0,0.7)$) {};
\node[circle, fill=gray,inner sep=0pt,minimum size=3pt] (Z4) at ($(Y4)+(0,0.7)$) {};

\draw[dashed,thin] (A4) -- (X4) -- (B4);
\draw[dashed,thin] (A4) -- (Y4) -- (B4);
\draw[dashed,thin] (A4) -- (Z4) -- (B4);

\node[circle, fill=black,inner sep=0pt,minimum size=3pt,label=below:{$f1$}] (f1) at ($(A4)+(1.2,-0.1)$) {};
\node[circle, fill=black,inner sep=0pt,minimum size=3pt,label=right:{$f2$}] (f2) at ($(X4)+(0,0.35)$) {};
\node[circle, fill=black,inner sep=0pt,minimum size=3pt,label=right:{$f3$}] (f3) at ($(Y4)+(0,0.3)$) {};

\draw[bend left=20] (f1) to (f2);
\draw[bend right=20] (f1) to (f2);
%\draw[bend left=20] (f1) to (f3);
%\draw[bend right=20] (f1) to (f3);
\draw (f1) .. controls ++(1.2,0) and ++(2.7,-1.9) ..  (f3);
\draw (f1) .. controls ++(-1.2,0) and ++(-2.7,-1.9) ..  (f3);
\draw[bend left=20] (f2) to (f3);
\draw[bend right=20] (f2) to (f3);

\node[] at ($(f1)+(0, -0.6)$) {$G^*$};
\node[] at ($(f1)+(0, -1.2)$) {$(d)$};
\end{tikzpicture}
\caption{(a) The graphic matroid $M(G)$ corresponding to graph $G$ is clearly connected; for any two distinct edges in $G$, a cycle contains them. (b) Deletion. The graph obtained by deleting edge $(B,X)$, the corresponding graphic matroid is disconnected. (c) Contraction. The graph obtained by contracting edge $(B,X)$, replacing vertices $B$ and $X$ with a new vertex $W$, and connecting the edges associated with $B$ and $X$ to $W$. The corresponding graphic matroid is connected. (d) Duality. The dual graph of $G$ with vertices $\{f_1, f_2, f_3\}$ and edges set consisting of a pair of parallel edges between each pair of vertices. The dual matroid of $M(G)$ is the  the corresponding graphic matroid of $G^*$, i.e.,$M(G)^*=M(G^*)$.}
\label{fig:matroid-concepts}
\end{figure}
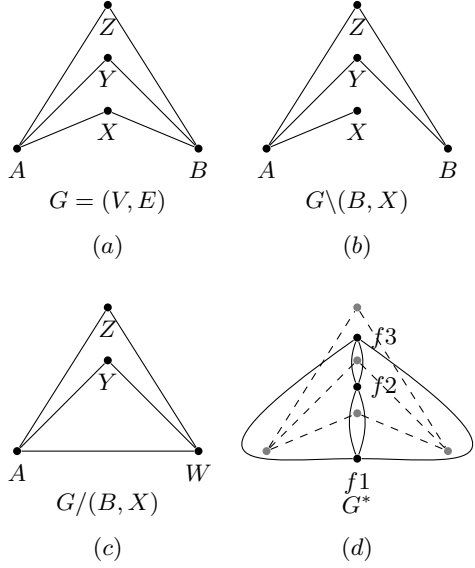

\begin{lemma}\label{le:minor-commute}
Let $M$ be a matroid on the finite ground set $[n]$, $e$ and $f$ are two distinct elements of $[n]$. Then
\begin{enumerate}
\item[(1)] $(M\backslash{e})\backslash{f}=(M\backslash{f})\backslash{e}$;
\item[(2)] $(M/e)/f=(M/f)/e$;
\item[(3)] $(M\backslash{e})/f=(M/f)\backslash{e}$.
\end{enumerate}
\end{lemma}
The Lemma~\ref{le:minor-commute} is easy to verify and the proof can be found in \cite{OUP/oxley11,welsh1976matroid}. If $T\subseteq [n]$, we write $M\backslash{T}$ and $M/T$ as successive deletions or contractions of the elements in $T$, respectively.

Deletion and contraction are dual operations, in the sense that, for any $e\in[n]$,
\begin{equation*}
 M/e=(M^*\backslash{e})^*\;\;\text{and}\;\;M\backslash{e}=(M^*/e)^*.
\end{equation*}

\begin{definition}[\textbf{Minor}]
Let $M$ be a matroid, a \emph{minor} of $M$ is any matroid that can be obtained from $M$ by a sequence of deletions and contractions. 
\end{definition}

\section{Quantum states and matroids}\label{main-section}
In this section, we will prove our results and investigate the properties of quantum states endowed with matroid structures, including quantum measurement and state flipping (dual states). We begin by introducing multiaffine polynomials, which serve as a fundamental bridge between matroid theory and the representation of quantum states.

\subsection{Multiaffine polynomials}
A \emph{multiaffine polynomial} (or \emph{multilinear polynomial}) is a multivariate polynomial in which each variable appears with degree at most one; equivalently, no variable is raised to a power greater than one, and each monomial is a product of distinct variables. A multivariate polynomial is \emph{homogeneous} if every monomial of the polynomial has  the same degree.

The following lemma 
%in  Choe, Oxley, Sokal and Wagner~\cite{choe2004homogeneous} 
for multiaffine polynomial is obvious but crucial.

\begin{lemma}[\cite{choe2004homogeneous}, Lemma 4.7]\label{le:multiaffine-product}
Let $p_1(x)$ and $p_2(x)$ be non-constant polynomials in the variables $x=\{x_i\}_{i\in[n]}$ with complex coefficients. If $p_1p_2$ is multiaffine, then
\begin{enumerate}
    \item[(a)] There exists a partition $\{E_1,E_2\}$ of $[n]$ such that $p_i$ uses only the variables $\{x_e\}_{e\in E_i}$ $\;(i=1,2)$.
    \item[(b)] $p_1$ and $p_2$ are both multiaffine.\\
\end{enumerate}
\end{lemma}

\begin{definition}
Let $p(x)=\sum_{S\subseteq[n]}\alpha_{S}x^S\in\mathbb{C}[x_1,\dots,x_n]$ be a polynomial, the support of $p(x)$ is defined as 
\begin{equation*}
    \supp(p) = \{S\subseteq[n]|\alpha_S\neq 0\}.
\end{equation*}
\end{definition}

\begin{lemma}[\cite{choe2004homogeneous}, extend Proposition 4.6 to complex coefficients]\label{le:conncted-matroid-to-irreducible}
Let $p$ be a multiaffine polynomial in the variables $\{x_i\}_{i\in[n]}$  with complex coefficients. If $\supp(p)$ is the collection of bases of a connected matroid on the ground set $[n]$, then $p$ is irreducible.
\end{lemma}

In the literature~\cite{choe2004homogeneous}, the original proposition is stated for real coefficients; however, the result holds over complex coefficients without requiring any changes to the proof. The proofs of Lemma~\ref{le:multiaffine-product} and Lemma~\ref{le:conncted-matroid-to-irreducible} can be found in ~\cite{choe2004homogeneous}.

Let $|\psi\rangle=\sum_{S\subseteq[n]}\alpha_S|S\rangle$ be a quantum state of an $n$-qubit system, we assign a polynomial in the variables $x=\{x_i\}_{i\in [n]}$ to $|\psi\rangle$,
\begin{equation}
    g_{|\psi\rangle}(x)=\sum_{S\in\supp(|\psi\rangle)}\alpha_Sx^S.
\end{equation}
We call the polynomial $g_{|\psi\rangle}$  the  \emph{generating polynomial} of the quantum state $|\psi\rangle$. Obviously, $g_{|\psi\rangle}$ is a multiaffine polynomial.
For the matroid-supported states, since the bases of the matroid are all of equal size, the generating polynomials of the matroid-supported state are homogeneous.
Now, the GME of a quantum state can be characterized by the irreducibility of its generating polynomial.

\begin{lemma}\label{le:entanglement-by-irreducibility}
A quantum state is genuinely entangled if and only if its generating polynomial is irreducible.    
\end{lemma}
\begin{proof}
 We prove its contrapositive. Let $|\psi\rangle=\sum_{S\subseteq[n]}\alpha_S|S\rangle$ and $g_{|\psi\rangle}(x)$ be a quantum state and its generating polynomial, respectively.

($\Leftarrow$). If $|\psi\rangle$ is separable, that is, $[n]$ has a partition $\{A,B\}$ such that
 \begin{equation}
     \begin{array}{rcl}
        |\psi\rangle&=&\sum_{S\in\supp{|\psi\rangle}}\alpha_S|S\rangle\\
        &=&|\psi_A\rangle\otimes|\psi_B\rangle\\
        &=&(\sum_{S_A\in\supp(|\psi_A\rangle}\alpha_{S_A}|S_A\rangle)\otimes\\
        &&(\sum_{S_B\in\supp(|\psi_B\rangle}\alpha_{S_B}|S_B\rangle)\\
        &=&\sum_{S_A\in\supp(|\psi_A),S_B\in\supp(|\psi_B)}\alpha_{S_A}\alpha_{S_B}|S_A\sqcup S_B\rangle,
     \end{array}
 \end{equation}
 where $|\psi_A\rangle$ and $|\psi_B\rangle$ are the quantum states of the qubits determined by $A$ and $B$, respectively.

Now, the generating polynomial of $|\psi\rangle$ can be written as
\begin{equation}
    \begin{array}{rcl}
        g_{|\psi\rangle}(x) &=&\sum_{S\in\supp(|\psi\rangle)}\alpha_Sx^S\\
        &=&\sum_{S_A\in\supp(|\psi_A),S_B\in\supp(|\psi_B)}\alpha_{S_A}\alpha_{S_B}x^{S_A\sqcup S_B}\\
        &=&(\sum_{S_A\in\supp(|\psi_A\rangle}\alpha_{S_A}x^{S_A})\;\cdot\\
        &&(\sum_{S_B\in\supp(|\psi_B\rangle}\alpha_{S_B}x^{S_B})\\
        &=&g_{|\psi_A\rangle}(x)\cdot g_{|\psi_B\rangle}(x).
    \end{array}
\end{equation}
So, a biseparable quantum state implies that its generating polynomial is reducible.\\

($\Rightarrow$).On the other hand, if the generating polynomial of $|\psi\rangle$ is reducible, by Lemma~\ref{le:multiaffine-product}, there exists a partition $\{A,B\}$ of $[n]$ and multiaffine polynomials $f\in\mathbb{C}[x_i:i\in A]$, $h\in\mathbb{C}[x_i:i\in B]$ such that 
\begin{equation*}
    g_{|\psi\rangle} = f\cdot h.
\end{equation*}
Let $f(x)=\sum_{S_A\in\supp(f)}\alpha_{S_A}x^{S_A}$ and $h(x)=\sum_{S_B\in\supp(h)}\alpha_{S_B}x^{S_B}$. Then 
\begin{equation}
    \begin{array}{rcl}
        g_{|\psi\rangle}(x)&=&\sum_{S\in\supp(|\psi\rangle)}\alpha_Sx^{S}\\
        &=&(\sum_{S_A\in\supp(f)}\alpha_{S_A}x^{S_A})\;\cdot\\
        &&(\sum_{S_B\in\supp(h)}\alpha_{S_B}x^{S_B})\\
        &=& \sum_{S_A\in\supp(f),S_B\in\supp(h)}\alpha_{S_A}\alpha_{S_B}x^{S_A\sqcup S_B}  
    \end{array}
\end{equation}
It is not difficult to see that $\supp(|\psi\rangle)=\{S_A\sqcup S_B|S_A\in\supp(f),S_B\in\supp(h)\}$, that is, for any $S\in\supp(|\psi\rangle)$, there exists a unique pair $(S_A,S_B)$ such that $S=S_A\sqcup S_B$ and $\alpha_S=\alpha_{S_A}\alpha_{S_B}$. So
\begin{equation}
    \begin{array}{rcl}
        |\psi\rangle&= & \sum_{S\in\supp{|\psi\rangle}}\alpha_S|S\rangle \\
         &=& \sum_{S_A\in\supp(f)}\sum_{S_B\in\supp(h)}\alpha_{S_A}\alpha_{S_B}|S_A\sqcup S_B\rangle\\
         &=&(\sum_{S_A\in\supp(f)}\alpha_{S_A}|S_A\rangle)\;\otimes\\
         &&(\sum_{S_B\in\supp(h)}\alpha_{S_B}|S_B\rangle).
    \end{array}
\end{equation}
Thus, from the reducible generating polynomial of $|\psi\rangle$, we can deduce that $|\psi\rangle$ is biseparable.
\end{proof}

\subsection{Quantum entanglement and matroid connectivity}
Matroid connectivity captures the indivisible structure of a combinatorial system—whether a graph without bridges or a matrix whose column space resists direct-sum decomposition—by ensuring that every pair of elements coexists in a common circuit, preventing nontrivial rank-preserving partitions. The property of matroid connectivity reminds us of the characterization of quantum entanglement, where a composite system resists factorization into independent subsystems. Both notions quantify irreducibility: circuits link matroid elements just as quantum correlations link subsystems. This idea paves the way for the cross-fertilization of combinatorial optimization and quantum information theory.

Specifically, for a matroid-supported state, its entanglement can be understood through the underlying matroid's connectivity. While establishing this correspondence directly from the formal definition of connected matroid is nontrivial, Lemma~\ref{le:connected-matroid} and Lemma~\ref{le:connected-matroid-sum} providing a combinatorial framework for understanding quantum entanglement. A circuit acts like an invisible chain linking each pair of qubits. If the underlying matroid of a matroid-supported state is connected — that is, for every pair of qubits there exists a circuit that contains both — then there is no partition of all these qubits that would allow the matroid-supported state to factorize. Otherwise, there exists a pair of qubits for which no circuit contains both. We obtain the following theorem.

%In the following, we restate Theorem~\ref{th:matroid-supported-state}.

\begin{theorem}\label{th:matroid-supported-state}
 Let $|\psi\rangle=\sum_{S\subseteq [n]}\alpha_S|S\rangle$ be a matroid-supported state with supporting matroid $M$. If $M$ is a connected matroid, then $|\psi\rangle$ is genuinely entangled.   
\end{theorem}
\begin{proof}
 Let $g_{|\psi\rangle}(x)=\sum_{S\in\supp{|\psi\rangle}}\alpha_Sx^S$ be the generating polynomial of $|\psi\rangle$. Since $|\psi\rangle$ is a matroid-supported state with supporting matroid $M$. Thus $\supp(g_{|\psi\rangle})$ is the collection of bases of $M$.
 By Lemma~\ref{le:conncted-matroid-to-irreducible}, we know that $g_{|\psi\rangle}$ is irreducible if $M$ is a connected matroid. Then applying Lemma~\ref{le:entanglement-by-irreducibility}, we conclude that $|\psi\rangle$ is genuinely entangled.
\end{proof}

%Todo::\remove{The most general method for certifying GME is the entanglement witness (EW)~\cite{guhne2009entanglement}. An EW is a Hermitian operator $W$ such that its expectation value is non-negative for every biseparable state. Then a negative expectation value $\text{Tr}(W\rho)$ implies that the state $\rho$ is genuinely entangled. To detect matroid-supported states with different coefficients, we need to design different entanglement witnesses. However, by using our results, we only need to show that the corresponding matroid is connected. This greatly increases the detection efficiency. Since there exists an efficient algorithm~\cite{cunnigham1974a,bixby1979matroids} for determining whether a matroid is connected.}

The converse of Theorem~\ref{th:matroid-supported-state} does not hold in general. In particular, an entangled matroid-supported state does not necessarily correspond to a connected matroid. We now present a counter-example demonstrating this failure of the converse.

\begin{example}\label{ex:disconnected-matroid-shows-GE}
Let $|\psi\rangle = \alpha_{13}|1010\rangle+\alpha_{14}|1001\rangle+\alpha_{23}|0110\rangle+\alpha_{24}|0101\rangle$ be a 4-qubit quantum state with support $\{\{1,3\},\{1,4\},\{2,3\},\{2,4\}\}$, this is a matroid-supported state. But $|\psi\rangle$ is genuinely entangled if and only if 
\begin{equation*}
    \begin{vmatrix}
    \alpha_{13} & \alpha_{23}\\
    \alpha_{14} & \alpha_{24}
    \end{vmatrix} \neq 0.
\end{equation*}
But the support $\{\{1,3\},\{1,4\},\{2,3\},\{2,4\}\}=\{\{1\},\{2\}\}\otimes\{\{3\},\{4\}\}$ which means that the supporting matroid of $|\psi\rangle$ is separable.
\end{example}

As illustrated by Example~\ref{ex:disconnected-matroid-shows-GE}, entangled state can be constructed on disconnected matroid. This asymmetry reveals that the properties of an arbitrary matroid-supported state do not reliably characterize the structural properties of its underlying matroid.
To establish a more robust and bijective correspondence, we define a canonical state called  matroid-base state(see Definition~\ref{def:matroid-base-state}). This state is constructed as the uniform superposition over all bases of a matroid.

\begin{theorem}\label{th:matroid-base-state}
Let $M=([n],\mc{B})$ be a matroid defined on the finite ground set $[n]$, where $\mc{B}\subseteq 2^{[n]}$ is the collection of bases of $M$, the matroid-base state $|M\rangle=\frac{1}{\sqrt{|\mc{B}|}}\sum_{S\in\mc{B}}|S\rangle$ is genuinely entangled if and only if $M$ is a connected matroid.   
\end{theorem}

\begin{proof}
 The matroid-base state $|M\rangle$ is a matroid-supported state, by Theorem~\ref{th:matroid-supported-state}, the sufficient condition is proved.

 We now prove the necessary condition. Since $|M\rangle$ is genuinely entangled, by Lemma~\ref{le:entanglement-by-irreducibility}, its generating polynomial $g_{|M\rangle}$ is irreducible. 

 Assume that $M$ is disconnected, let $M=M_1\oplus M_2$. Then the multiaffine polynomial $g_{|M\rangle}$ can be rewritten as
 \begin{equation}
 \begin{array}{rcl}
     g_{|M\rangle} (x)&=&\frac{1}{\sqrt{|\mc{B}(M)|}}\sum_{B\in\mc{B(M)}}x^B\\
      &=&\frac{1}{\sqrt{|\mc{B}(M)|}}\sum_{(B_1,B_2)\in\mc{B}(M_1)\otimes\mc(B)(M_2)}x^{B_1\sqcup B_2}\\
      &=&\frac{1}{\sqrt{|\mc{B}(M)|}}(\sum_{B_1\in\mc{B}(M_1)}x^{B_1})\;\cdot\\
      &&(\sum_{B_{2}\in\mc{B}(M_2)}x^{B_2}).
 \end{array}
 \end{equation}
This contradicts the irreducibility of $g_{|M\rangle}$, and the assumption does not hold.

Hence, the matroid $M$ is connected, completing the proof of necessity.
\end{proof}

Although the quantum states $|\Psi^{\pm}\rangle$, $|W_n\rangle$ and $|D^n_k\rangle$ in Example~\ref{ex:support-set-state} are already known to be genuinely entangled, their GME can be elegantly elucidated within the framework of matroid theory by demonstrating that the support of each state corresponds to the collection of bases of a connected matroid.

\begin{corollary}\label{co:dicke}
The Bell state $|\Psi^\pm\rangle=\frac{1}{\sqrt{2}}(|10\rangle\pm|01\rangle)$, the $n$-qubit $\mc{W}$ state $|\mc{W}_n\rangle=\frac{1}{\sqrt{n}}(|10\cdots0\rangle+|010\cdots0\rangle+\cdots+|00\cdots01\rangle)$, or, more generally, the $(n,k)$ Dicke state $|D^{n}_{k}\rangle=\frac{1}{\sqrt{\binom{n}{k}}}\sum_{S\subseteq[n]\;|\;|S|=k}|S\rangle$ is a genuinely entangled state for $1\leq k\leq n-1$.
\end{corollary}

\begin{proof}
It is not difficult to see that the Bell state $|\Psi^\pm\rangle$ and $\mc{W}$ state $|W_n\rangle$ are special Dicke states. 

 Since the support set $\supp(|D^{n}_{k}\rangle)$ of $|D^{n}_{k}\rangle$ is the collection of bases  of the unifom matroid $U_{k,n}$. That is to say, $|D^{n}_{k}\rangle$ is a matroid-base state. For $k\neq0$ and $k\neq n$, the uniform matroid $U_{n,k}$ is connected. So, by Theorem~\ref{th:matroid-base-state}, we can prove that $|D^{n}_{k}\rangle$ is a genuinely entangled state.
\end{proof}

\subsection{Quantum measurement and matroid minor}

%A compelling structural parallel exists between the process of projective measurement in quantum mechanics and the minor operation in matroid theory. Both operations function as irreversible, structure-reducing processes that condition a global system on local information.

%In quantum mechanics, a projective measurement acts on a state in superposition by projecting its vector onto an eigenspace associated with a specific measurement outcome. This process collapses the state from a complex linear combination of all possible basis states into a new state confined to a subspace that is consistent with the observation.

%Analogously, the matroid minor operation reduces the complexity of a matroid by conditioning its structure on a subset of its ground elements.

%This reveals a formal correspondence: quantum measurement acts on a state by restricting it to a specific subspace, just as the minor operation acts on a matroid by restricting its family of bases. In both contexts, a global structure (a superposition of all bases or the full set of dependencies) is reduced to a more localized form that reflects newly acquired information.

%In this paper, we only consider measurement in the standard computational basis.

%First, let's look at the connection between quantum measurement and matroid minor in the simplest case --- measuring a qubit and deleting/contracting an element.

A compelling structural analogy can be drawn between quantum measurement and the minor operation in matroid theory.  Both processes are inherently irreversible, reduce the effective size or dimensionality of the system, and yield only partial information about the original structure.

In quantum mechanics, a measurement projects the quantum state onto a lower-dimensional subspace of the Hilbert space, thereby collapsing the state and discarding part of the global information. Analogously, taking a minor of a matroid—by deleting or contracting some elements—produces a new matroid of smaller rank, representing a reduction of the original dependency structure. 
From this perspective, both quantum measurement and matroid minors embody a form of structural restriction or dimensional reduction. In this section, we investigate the intrinsic correspondence between these two notions.

Throughout this discussion, we restrict our attention to the measurement in the $Z$-basis. 

We begin with the simplest case: measuring a single qubit and taking the corresponding matroid minor obtained by deleting or contracting one element.

\begin{lemma}\label{le:measurement-minor-1}
Let $|\psi\rangle=\sum_{S\subseteq2^{[n]}}\alpha_S|S\rangle$ be a matroid-supported state with supporting matroid $M=([n], \mc{B})$ and $e\in[n]$. If the qubit labeled by $e$ is being measured,
\begin{enumerate}
\item [(1)] If the measurement result is 0, then the quantum state collapses to a matroid-supported state with supporting matroid $M\backslash{e}$;
\item [(2)] If the measurement result is 1, then the quantum state collapses to a matroid-supported state with supporting matroid $M/{e}$.
\end{enumerate}
\end{lemma}

\begin{proof}
Given a $e\in[n]$, we partition $\mc{B}$ into two disjoint unions: $\mc{B}=\mc{B}_{\overline{e}}\;\sqcup\;\mc{B}_e$, where $\mc{B}_{\overline{e}}=\{B\in\mc{B}\;|\;e\notin B\}$ and $\mc{B}_{e}=\{B\in\mc{B}\;|\;e\in B\}$. We show that if $\mc{B}_{\overline{e}}$( $\mc{B}_e$) is nonempty, then $\mc{B}_{\overline{e}}$($\mc{B}_e$) is the collection of bases of a matroid on $[n]-e$. We only give proof for $\mc{B}_{\overline{e}}$, the proof for $\mc{B}_e$ is similar.

For any $B_1,B_2\in\mc{B}_{\overline{e}}$, if $x\in B_1 - B_2$, by the Base Axioms (Lemma~\ref{th:base-axioms}), there exists $y\in B_2-B_1$ such that $B_1-x+y\in\mc{B}$. Since $e\notin B_1$ and $e\notin B_2$, so $e\notin B_1-x+y$, meaning that $B_1-x+y\in\mc{B}_{\overline{e}}$. Applying the Base Axioms again, we know that $\mc{B}_{\overline{e}}$ does be the collection of bases of some matroid. In fact, this matroid is $M\backslash{e}$.

Similarly, we can prove that $\mc{B}_e$ is the collection of bases of $M/e$.

Now, we rewrite the state $|\psi\rangle$ as
\begin{equation}\label{eq:rewrite-psi}
    |\psi\rangle=a|0_e\rangle|\psi\backslash{e}\rangle+b|1_e\rangle|\psi/e\rangle,
\end{equation}
\noindent{}where $a=\sqrt{\sum_{S\in\mc{B}_{\overline{e}} }|\alpha_S|^2}$, $b=\sqrt{\sum_{S\in\mc{B}_e}|\alpha_S|^2}$, $|\psi\backslash{e}\rangle=\frac{1}{a}\sum_{S\in\mc{B}_{\overline{e}}}\alpha_S|S\rangle$ and $|\psi/e\rangle=\frac{1}{b}\sum_{S\in\mc{B}_e}\alpha_S|S-e\rangle$.

If we perform a measurement on the qubit labeled by $e$ in the state $|\psi\rangle$, by Equation~(\ref{eq:rewrite-psi}), we observe the followings:
\begin{itemize}
\item If the measurement outcome is $0$, the quantum state collapses to $|\psi\backslash{e}\rangle$ is still a matroid-supported state with supporting matroid $M\backslash{e}$ --- the minor of $M$ obtained by deleting $e$.
\item If the measurement outcome is $1$, the quantum state collapses to $|\psi/e\rangle$ is still a matroid-supported state with supporting matroid $M/e$ --- the minor of $M$ obtained by contracting $e$.
\end{itemize}
\end{proof}

From Lemma~\ref{le:measurement-minor-1}, we can see the natural correspondence between the result of measuring a single qubit and a single element  deletion/contraction of matroid. More generally, if we measure multiple qubits, what is the relationship (in the sense of matroids) between the collapsed quantum state and the original quantum state? We describe this relationship as the following theorem.

\begin{theorem}\label{th:meaurement-minor}
Let $|\psi\rangle$ be a matroid-supported state with supporting matroid $M=([n],\mc{B})$. Let $T_0$ and $T_1$ represent the qubits measured resulting 0 and 1 respectively, then the resulting quantum state (over the unmeasured qubits) is still a matroid-supported state supported by matroid $M\backslash{T_0}/T_1$, the minor of $M$ obtained by deleting $T_0$ and contracting $T_1$. (See Figure~\ref{fig:measurement-minor}.)
\end{theorem}

%\begin{turnpage}
\begin{figure}[htbp]
\centering
\begin{tikzpicture}
% 嵌套椭圆
\node[] (psi) at (0,1.8) {$|\psi\rangle$};
\node[] (psi1) at (4,1.8) {$|\psi'\rangle$};
\node[] (M) at (0,0) {$M$};
\node[] (M1) at (4,0) {$M\backslash{T_0}/T_1$};
\node[] (T) at (2,0.9) {$\{T_0, T_1\}$};

\draw[<->] (psi) -- (M);
\draw[<->] (psi1) -- (M1);
\draw[->] (psi) --node[pos=0.5, above]{Measurement} (psi1);
\draw[->] (M) -- node[pos=0.5, below] {Minor} (M1);

\draw[->] ($(psi)!0.5!(psi1)$) -- (T) ;
\draw[->] (T)-- ($(M)!0.5!(M1)$);
\end{tikzpicture}
  \caption{Quantum measurement and matroid minors. A matroid-supported state $|\psi\rangle$ supported by matroid $M$. $T_0$ and $T_1$ denote the qubits with measurement results of 0 and 1 in the $Z$-basis, respectively. After measurement, the matroid-supported state collapse to $|\psi'\rangle$ is still a matroid-supported state supported by matroid $M\backslash{T_0}/T_1$.}
  \label{fig:measurement-minor}
\end{figure}
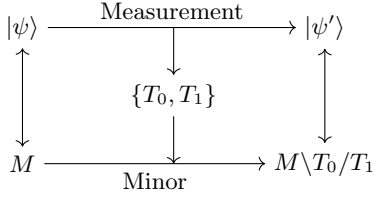
%\end{turnpage}

\begin{proof}
Let $T=T_0 \cup T_1$, when $|T|=1$, the results is as described in Lemma~\ref{le:measurement-minor-1}. We now use induction on the size of $T$. The correctness of our induction is based on the following two facts.
\begin{fact}\label{fa:deletion-contraction-commute}
The deletion and contraction of matroid commute with each other and with themselves. (Lemma~\ref{le:minor-commute})
\end{fact}
\begin{fact}\label{fa:measurement-commute}
If we know the initial state and the state after the measurement (on the unmeasured qubits), then measurement order (on the measured qubits) does not matter.
\end{fact}

Assume $|T|=t+1$ ($t>0$ be an integer),  let $T'\subset T$  be the first $t$ measured qubits and $e$ be the last measured qubit (Note that he order of these measurements does not matter, and its plausibility stems from Fact~\ref{fa:measurement-commute}). By the induction hypothesis, a local measurement in the $Z$-basis determined by $T'$, the resulting matroid-supported state $|\psi'\rangle$ has $M\backslash{}T'_0/T'_1$ as the supporting matroid, where $T'_0$ and $T'_1$ are disjoint partitions of $T'$ represent the qubits measured resulting 0 and 1, respectively.

When we continue to measure the qubit labeled by $e$, it is equivalent to measuring a single  qubit in the initial state $|\psi'\rangle$. By Lemma~\ref{le:measurement-minor-1},
\begin{itemize}
\item If the measurement result is $0$, the quantum state collapses to a matroid-supported state with supporting matroid $(M\backslash{}T'_0/T'_1)\backslash{e}$. By the Fact~\ref{fa:deletion-contraction-commute}, the associated matroid is $M\backslash{}(T'_0+e)/T'_1$.
\item Similarly, if the measurement result is $1$, the quantum state collapses to a matroid-supported state with supporting matroid $M\backslash{}T'_0/(T'_1+e)$.
\end{itemize}
\end{proof}

\subsection{Dual states}
Duality is a fundamental and versatile concept that reveals the underlying symmetry of mathematical and physical structures. In quantum mechanics, duality appears in many forms, such as wave–particle duality, the Fourier duality between position and momentum. These dual formulations often expose hidden symmetries, facilitate structural analysis, and provide alternative yet equivalent perspectives on the same physical phenomena. Especially in matroid theory, duality is a core concept. Inspired by this, we define the dual state of a quantum state.

\begin{definition}[Dual state]\label{def:dual-quantum-state}
Given any quantum state $|\psi\rangle = \sum_{S\subseteq[n]}\alpha_S|S\rangle$, we define its dual state as
\begin{equation}
    |\psi\rangle^*=X^{\otimes n}|\psi\rangle=\sum_{S\subseteq[n]}\alpha_S|\overline{S}\rangle,
\end{equation}
where $X$ is the Pauli $X$ gate.
\end{definition}

\begin{definition}[Self-dual state]
For any quantum state $|\psi\rangle$, we call $|\psi\rangle$  a \emph{self-dual state}
if $|\psi\rangle^*=|\psi\rangle$.
\end{definition}
Self-dual states have good symmetry and often appear in quantum information and quantum computing, such as Bell states, GHZ states, and uniform superposition states (of all computational bases).

\begin{proposition}\label{le:dual-state-engtangled}
For any quantum state $|\psi\rangle$, $|\psi\rangle$ is genuinely entangled if and only if $|\psi\rangle^*$ is genuinely entangled.
\end{proposition}

\begin{proof}
Since Pauli $X$ is a single qubit gate, it does not change the GME of a quantum state. According to the definition of the dual state, it has the same GME as the original state.   
\end{proof}

Alternatively, the result can be established by analyzing the generating polynomial of the quantum state. Since for $S\subseteq [n]$, $x^{\overline{S}}=x^{[n]}\cdot(\frac{1}{x})^S$. Note that for an n-tuple $x=(x_1,\dots,x_n)$, $\frac{1}{x}$ represents $(\frac{1}{x_1},\cdots,\frac{1}{x_n})$, and $(\frac{1}{x})^S=\prod_{i\in S}\frac{1}{x_i}$. It is easy to obtain the generating polynomial of $|\psi\rangle^*$,
\begin{equation}
g_{|\psi\rangle^*}(x)=x^{[n]}p_{|\psi\rangle}(\frac{1}{x}).
\end{equation}
And it has the same reducibility as $g_{|\psi\rangle}$.

\textbf{Remark}. In fact, the dual state we define is local unitary (LU) equivalent to the original state. Therefore, Proposition~\ref{le:dual-state-engtangled} is a well-known result \cite{nielsen2010quantum,horodecki2009quantum,monras2011entanglement}. We revisit this result in order to underscore the fundamental role played by duality.

For matroid-supported states, the properties of their dual states can be systematically explained through the powerful framework of matroid duality. This provides a direct mapping from quantum operations on the state to well-understood combinatorial operations on the underlying matroid.

\begin{proposition}\label{le:dual-matroid-supported-state}
If $|\psi\rangle$ is a matroid-supported state with supporting matroid $M$, then $|\psi\rangle^*$ is a matroid-supported state with supporting matroid $M^*$, the dual of $M$.
\end{proposition}

\begin{proof}
 Let $|\psi\rangle = \sum_{S\in\mc{B}(M)}\alpha_S|S\rangle$ be an $n$-qubit matroid-supported state, by the duality of matroids (see Lemma~\ref{le:matroid-duality}), $|\psi\rangle$ can be rewritten as
 \begin{equation}
     |\psi\rangle=\sum_{S\in\mc{B}(M*)}\alpha_{\overline{S}}|\overline{S}\rangle.
 \end{equation}
 By the definition of dual quantum state (see Definition~\ref{def:dual-quantum-state},
 \begin{equation}
     |\psi\rangle^*=X^{\otimes n}(\sum_{S\in\mc{B}(M^*)}\alpha_{\overline{S}}|\overline{S}\rangle)=\sum_{S\in\mc{B}(M^*)}\alpha_{\overline{S}}|S\rangle.
 \end{equation}
 It is straightforward to see that $|\psi\rangle^*$ is a matroid-supported state with supporting matroid $M^*$.
\end{proof}

\begin{proposition}\label{le:dual-matroid-base-state}
Let $|\psi_M\rangle$ and $|\psi_{M^*}\rangle$ are matroid-base states associated to matroids $M$ and its dual matroid $M^*$, respectively. Then $|\psi_M\rangle^*=|\psi_{M^*}\rangle$.   
\end{proposition}
\begin{proof}
Since $|\psi_M\rangle$ is a matroid-base state, let $|\psi_M\rangle=\frac{1}{\sqrt{|\mc{B}(M)|}}\sum_{S\in\mc{B}(M)}|S\rangle$. So
\begin{equation}
\begin{array}{rcl}
 |\psi_M\rangle^*&=&X^{\otimes n}(\frac{1}{\sqrt{|\mc{B}(M)|}}\sum_{S\in\mc{B}(M)}|S\rangle)\\
 &=&\frac{1}{\sqrt{|\mc{B}(M)|}}\sum_{S\in\mc{B}(M)}|\overline{S}\rangle\\
 &=&\frac{1}{\sqrt{|\mc{B}(M^*)|}}\sum_{S\in\mc{B}(M^*)}|S\rangle\\
 &=&|\psi_{M^*}\rangle.
\end{array}
\end{equation}
\end{proof}

\section{Conclusions and discussions}\label{conclusion}
\subsection{Conclusions}
In this work, we investigate quantum states endowed with matroid structures. We link quantum entanglement to matroid connectivity, offering a combinatorial way to characterize entanglement. For general matroid-supported states, a connected matroid guarantees entanglement; for matroid-base states, connectivity is both necessary and sufficient. We also connect local $Z$-basis measurements to matroid operations: result 0 acts as deletion, result 1 as contraction. Thus, measurements yield matroid minors, and quantum state duality mirrors matroid duality. These parallels reveal deep structural symmetry and give a unified lens for analyzing quantum states. Overall, matroid theory offers a powerful conceptual tool for understanding entanglement, measurement, and duality in quantum systems.

\subsection{Discussion}
Since matroids have a wide range of applications, matroid-supported states possess rich combinatorial structures that can be exploited in quantum algorithm design, combinatorial optimization, and the quantum simulation of discrete systems. 
Although we have not yet found concrete applications of matroid supported states, the following two issues warrant further investigation: 

\begin{enumerate}
\item \textbf{Matroid-supported states preparation}.

Quantum state preparation is a fundamental task in quantum information and quantum computing~\cite{shende2005synthesis,sun2023asymptotically,yuan2023optimal,zhang2022quantum,li2025nearly}. 
An intriguing question in this context is whether {matroid-supported quantum states} can be prepared efficiently. Some special cases of this problem are already well understood.  For instance, {Dicke states}, which correspond to {uniform matroids}, have been extensively studied due to their broad applications in quantum information processing, and several efficient preparation methods are known~\cite{bartschi2019deterministic,yuan2025depth}. Based on this, the class of
matroid-base states  supported by partition matroids \cite{recski1973partitional} (an important class of matriods) can be efficiently prepared by constructing tensor products of Dicke states, since  partition matroids are direct sums of uniform matroids.

A particularly relevant family of states is that of {Hamming-weight-preserving states}~\cite{raveh2024deterministic,farias2025quantum,mao2024towards}, in which all computational basis components share the same Hamming weight $k$. The family of {matroid-supported states} naturally arises as a subclass of these states. Efficient preparation of Hamming-weight-preserving states has recently attracted considerable attention \cite{li2025preparation,luo2025optimal}.
An interesting problem, however, is whether the intrinsic structural properties of matroids---such as the basis-exchange axiom---can be further exploited to reduce the circuit depth, gate count, or ancillary qubit number of the quantum circuit for preparing such a state. Addressing this problem would be highly valuable, particularly for realizing quantum walks on matroid basis-exchange graphs.

\item \textbf{Quantum walks on the basis-exchange graphs of matroids}.
Quantum walks provide a general framework for quantum algorithm design, serving as the quantum analogue of classical random walks through the use of superposition, interference, and entanglement~\cite{aharonov2001quantum,ambainis2003quantum,venegas2012quantum,li2026deterministic, wang2025unifying, li2025unbounded, li2025derandomization, xu2022robust}.
%(Aharonov, Ambainis, Kempe and Vazirani~\cite{aharonov2001quantum}, Ambainis~\cite{ambainis2003quantum} and Venegas-Andraca~\cite{venegas2012quantum}). 
They have been shown to yield exponential speedups for certain computational problems~\cite{childs2003exponential,li2024recovering}.
%(Childs, Cleve, Deotto, Farhi, Gutmann and Spielman~\cite{childs2003exponential} and Li, Li and Luo~\cite{li2024recovering}). 

On the classical side, a random walk on a matroid basis-exchange graph defines a Markov process on the set of all bases of a matroid, where edges connect pairs of bases that differ by a single element exchange, in accordance with the matroid basis-exchange axiom. This graph-theoretic formulation captures the combinatorial geometry of matroids and provides a natural foundation for algorithmic analysis. Such walks have found wide applications, including approximate uniform sampling of bases, counting the total number of bases, and estimating network reliability~\cite{feder1992balanced,anari2018log,anari2019log,anari2020isotropy,chen2024near}.

Consequently, quantum walks on the matroid basis-exchange graphs constitute a promising research direction, offering both theoretical richness and significant potential for quantum algorithmic applications.     
\end{enumerate}

\begin{acknowledgments}
This work was supported by the National Key Research and Development Program of China (Grant No.2024YFB4504004), the National Natural Science Foundation of China (Grant No. 92465202, 62272492), the Guangdong Provincial Quantum Science Strategic Initiative (Grant No. GDZX2303007, GDZX2503001), the Guangzhou Science and Technology Program (Grant No. 2024A04J4892).
\end{acknowledgments}

\bibliography{refs}% Produces the bibliography via BibTeX.

\end{document}